\documentclass[a4paper,10pt,oneside]{article}

\usepackage{amsmath}
\usepackage{amssymb}
\usepackage{amsthm}
\usepackage{enumerate}

\newcommand{\N}{\mathbb{N}}
\newcommand{\R}{\mathbb{R}}
\newcommand{\C}{\mathbb{C}}
\newcommand{\K}{\mathbb{K}}

\newcommand{\M}[1]{\mathcal{M}_{#1}}   
\newcommand{\MSa}[1]{\mathcal{M}_{#1,\mathrm{sa}}} 
\newcommand{\MP}[1]{\mathcal{M}_{#1}^{+}}   % nxn positive matrices
\newcommand{\MPP}[1]{\mathcal{M}_{#1}^{++}} % nxn invertible positive matrices
\newcommand{\DN}[1]{\mathcal{D}_{#1}} % density matrices

\newcommand{\di}{\displaystyle}
\newcommand*{\eig}[1]{\mathrm{Eig}\left(#1\right)}
\newcommand*{\tr}{\mathop{\textrm{Tr}}\nolimits}
\newcommand*{\Ran}{\mathop{\textrm{Ran}}\nolimits}
\newcommand*{\Id}{\mathop{\textrm{I}}\nolimits}
\newcommand{\eps}{\varepsilon}
\newcommand{\Ch}{\mathcal{C}}
\newcommand{\bigslant}[2]{
  {\raisebox{.2em}{$#1$}\left/\raisebox{-.2em}{$#2$}\right.}} 

\newcommand{\sep}{\mathrm{sep}}
\newcommand{\supp}{\mathrm{supp}}
\newcommand{\ent}{\mathrm{ent}}
\newcommand{\Ker}{\mathrm{Ker}}
\newcommand{\Hilb}{\mathrm{H}}
\newcommand{\Tra}{\mathrm{T}}

\newcommand*{\gz}[1]{\left(#1\right)}
\newcommand*{\sz}[1]{\left[#1\right]}
\newcommand*{\kz}[1]{\left\{#1\right\}}
\newcommand*{\scal}[1]{\left\langle#1\right\rangle}

\makeatletter
\let\@afterindentfalse\@afterindenttrue
\@afterindenttrue
\makeatother    

\newtheorem{definition}{Definition}
\newtheorem{theorem}{Theorem}
\newtheorem{lemma}{Lemma}

\title{On the Notion of Quantum Copulas\thanks{
keywords: copula, quantum state, quantum channel, entanglement;
MSC2010: 81P16, 81P40.
}}

\author{
Attila Lovas\thanks{lovas@math.bme.hu}, 
Attila Andai\thanks{andaia@math.bme.hu},\\
Department for Mathematical Analysis, \\
Budapest University of Technology and Economics,\\
Stoczek u. 2, Budapest, H-1521, Hungary}

\date{\today}

\begin{document}

\maketitle

\begin{abstract}
Working with multivariate probability distributions Sklar introduced 
  the notion of copula in 1959, which turned out to be a key concept
  to understand the structure of distributions of composite systems.  
Roughly speaking Sklar proved that a joint distribution can be represented 
  with its marginals and a copula.
The main goal of this paper is to present a quantum analogue of the notion 
  of copula. 
Our main theorem states that for any state of a composite quantum system there 
  exists a unique copula such that the sate and the copula are connected to 
  each other by invertible matrices, moreover, they are both separable or both
  entangled.
So considering copulas instead of states is a separability preserving
  transformation and efficiently decreases the dimension of the 
  state space of composite systems.
The method how we prove these results draws attention to the fact 
  that theorem for states can be achieved by considering the states as quantum 
  channels and using theorems for channels and different kind of positive maps.
\end{abstract}

\section{Introduction}

A famous problem was proposed by Fr\'echet in 1951
  about multivariate probability distributions \cite{Frechet1951}.
Sklar managed to obtain an outstanding result in Fr\'echet's problem
  by introducing the notion of copula in 1959 \cite{Sklar1959},
  which is known as Sklar's theorem.
Roughly speaking Sklar proved that a joint distribution can be represented 
  with its marginals and a copula.
It means that the copula has all the dependence information of 
  random variables \cite{Nelsen2006}.

First let us recall the Sklar's theorem in the classical two dimensional 
  setting.
The theorem states that every multivariate distribution function
  $H:\R^{2}\to\sz{0,1}$ of a random vector $(X_{1},X_{2})$,
  that is
\begin{equation}
\label{eq:intro:Hdef}
 H(x_{1},x_{2})
  =\mathrm{Pr}\gz{X_{1}\leq x_{1},X_{2}\leq x_{2}}, 
\end{equation}
  can be expressed in terms of its marginals 
  $F_{i}(x_{i})=\mathrm{Pr}\gz{X_{i}\leq x_{i}}$
  and a copula function $C:\sz{0,1}^{2}\to\sz{0,1}$, such that
\begin{equation}
\label{eq:intro:coupolafunction}
H(x_{1},x_{2})=C\gz{F_{1}(x_{1}),F_{2}(x_{2})}
\end{equation}
  holds, moreover, the copula function is unique on the set
  $\Ran F_{1}\times\Ran F_{2}$.
So in the continuous case the copula function is unique.

In noncommutative probability theory the above mentioned notions and
  constructions are difficult to work with, but the idea behind Sklar's
  theorem can be adopted.
In quantum mechanical setting, an $n$-level quantum system can be described
  on an $n$ dimensional Hilbert space and the state space $\DN{n}$ of this 
  quantum system  can be identified by the set of self adjoint, 
  positive semi definite $n\times n$ matrices with trace one.
Quantum mechanical axioms postulate \cite{Neumann1932} that 
  the Hilbert space of a composite system is the tensor product of the 
  Hilbert spaces associated with the components.
If a composite quantum system consists of two subsystems
  with state spaces $\DN{n}$ and $\DN{m}$
  then the state space of the composite system is $\DN{mn}$. 
Transferring the idea of Sklar's theorem to quantum setting one finds that
  the role of classical distribution function (Eq. \ref{eq:intro:Hdef})
  is played by an element $\rho\in\DN{mn}$ and the marginals 
  can be identified by states generated by partial trace, namely
  by $\tr_{1}\rho\in\DN{m}$ and $\tr_{2}\rho\in\DN{n}$.
Equation \eqref{eq:intro:coupolafunction} suggests to interpret 
  the quantum copula as a state $C_{\rho}\in\DN{mn}$ with uniform marginals.
It means that $\tr_{1}C_{\rho}$ and $\tr_{2}C_{\rho}$ should be the most 
  mixed states, which can be written as $\tr_{1}C_{\rho}=\frac{1}{m}\Id_{m}$ 
  and  $\tr_{2}C_{\rho}=\frac{1}{n}\Id_{n}$, where $\Id_{n}$ denotes the
  $n\times n$ identity matrix.
We will call these states to precopulas and denote their set by $\Ch'_{mn}$.

We use the concept of Hilbert's pseudometric and the corresponding
  Birkhoff--Hopf theorem for strictly positive quantum channels to 
  construct a copula from a given state.
It turns out that for any state $\rho\in\DN{nm}$ there are unique matrices
  $A$ and $B$ up to multiplication by unitaries (i. e. one can consider
  $U_{0}A$ and $U_{1}B$ matrices too with unitaries $U_{0}$ and $U_{1}$)
  such that $(A\otimes B)\rho(A\otimes B)^{*}$ is a precopula.
It means that there is an extra unitary like freedom in the precopula 
  construction to a given state. 
This freedom induces an equivalence relation $\sim$ on precopulas.
We introduce the notion of copula in quantum setting as an equivalence
  class of precopulas.
The set of copulas is defined as $\Ch_{mn}=\bigslant{\Ch'_{mn}}{\sim}$.

We can give an affirmative answer for the question of uniqueness of copulas in
  the following way.
States $\rho,\rho'\in\DN{mn}$ are called to be connected  
 if there exist invertible matrices $A\in\M{n}$ and $B\in\M{m}$ such that
$\rho'=(A^{*}\otimes B^{*})\rho (A\otimes B)$.
Our main theorem states that for any state $\rho\in\DN{mn}$ there 
 exists a unique copula $\tilde{\chi}\in\Ch_{mn}$ such that $\rho$ and 
 any representative $\chi$ of $\tilde{\chi}$
 are connected, moreover, $\chi$ is separable if and only if $\rho$ is 
 separable.
So considering copulas instead of states is a separability preserving
  transformation, moreover the set of copulas are considerably lower dimensional 
  than the set of states and thus for bipartite systems, we can give a 
  simplified but exact picture of the separable-entangled structure of the state 
  space. 
For instance, in the theory of quantum steering, which is a type of quantum 
  correlation intermediate between entanglement and Bell nonlocality, the copula 
  of a two-qubit state is proved to be useful because its "steering" properties 
  are the same as for the state itself \cite{Bowles2016,Milne2014,shuming2016}.

\section{Preliminary lemmas and notations}

We use $\M{n}$ to denote $n\times n$ matrices, $\MSa{n}$ to denote the  
  self-adjoint elements of $\M{n}$, $\MP{n}$ the cone of self-adjoint positive 
  semi-definite ones \cite{NielsenChuang,Neumann1932,Petz2008}. 

The state space of an $n$-level quantum system arises as the intersection of 
  $\MP{n}$ and the hyperplane of trace one matrices, that is
\begin{equation*}
\overline{\DN{n}}=\kz{\rho\in\MP{n} \mid \tr (\rho)=1}.
\end{equation*}
Invertible elements in $\MP{n}$ are denoted by $\MPP{n}$ and 
  $\DN{n}$ stands for the set of invertible density matrices.
If we would like to emphasize the underlying field, then we write $\M{n}(\K)$, 
  $\MSa{n}(\K)$, $\MP{n}(\K)$, $\MPP{n}(\K)$, etc., where $\K=\R,\C$. 
The notation $\eig{A}$ stands for the set of eigenvalues of an element  
  $A\in\M{n}$.

Assume that a composite quantum system consists of two subsystems  
  with state spaces $\DN{n}$ and $\DN{m}$ respectively. 
According to the axioms of quantum mechanics, the state space of the composite 
  system is $\DN{mn}$ \cite{Neumann1932,Petz2008}. 
Product states in $\DN{mn}$ are of the form $\rho^1\otimes\rho^2$, where 
  $\rho^1\in\DN{n}$ and $\rho^2\in\DN{m}$.

\begin{definition}
A state $\rho\in\DN{mn}$ is \emph{separable} 
  (or \emph{classically correlated}), if it can be written as a convex 
  combination of product states i.e.
\begin{equation*}
\rho = \sum_{i=1}^r p_i \rho_i^1\otimes\rho_i^2,
\end{equation*} 
  where $\rho_i^1\in\DN{n}$, $\rho_i^2\in\DN{m}$and $\gz{p_i}_{1\le i\le r}$ 
  is a probability vector.
\end{definition}

Non-separable states are called entangled. 
We use the notations $\DN{mn}^\sep$ and $\DN{mn}^\ent$ for the set of separable 
  and entangled quantum states, respectively. 
Roughly speaking, separability means that the subsystems are in the state 
  $\rho_i^1$ and $\rho_i^2$ with probability $p_{i}$. 
The information theoretical aspect of entanglement observed by Schr\"odinger 
  himself: "Best possible knowledge of the whole does not include the best 
  possible knowledge of its parts." \cite{Petz2008}.   

\begin{definition}
A linear map $\Phi:\M{n}\to\M{m}$ is said to be
\begin{enumerate}[i.]
\item \emph{positive} if $\Phi (\MP{n})\subseteq\MP{m}$,
\item \emph{strictly positive} or \emph{positivity improving} if 
  $\Phi (\MP{n}\setminus\kz{0})\subseteq\MPP{m}$,
\item \emph{completely positive} if the map 
  $\Id_{\M{k}}\otimes \Phi: \M{k}\otimes\M{n}\to\M{k}\otimes\M{m}$ is positive 
  for all $k\in\N$, where $\Id_{\M{k}}$ is the $\M{k}\to\M{k}$ identity map.
\end{enumerate}
\end{definition}

Let us introduce the matrix units $(E_{ij})_{i,j=1,\dots,n}$, where 
  $E_{ij}$ denotes the matrix whose entries are all $0$ except in the $ij$-th 
  cell, where it is $1$.
One can associate to every linear map $\Phi:\M{n}\to\M{m}$ an element of 
  $\M{n}\otimes\M{m}\cong\M{mn}$ in the following way
\begin{equation*}
\rho_\Phi=\sum_{i,j=1}^{n} E_{ij}\otimes\Phi (E_{ij}).
\end{equation*}

According to Choi's theorem on completely positive maps, the map 
  $\Phi\mapsto\rho_{\Phi}$ establishes a bijection between completely positive 
  maps and positive semi-definite matrices. This correspondence is called 
  Choi--Jamio\l{}kowski isomorphism and its inverse is denoted by 
  $\rho\mapsto\Phi_\rho$. 

\begin{theorem}[Choi--Jamio\l{}kowski 
  \cite{Choi,Jamiolkowski}]
A linear map $\Phi:\M{n}\to\M{m}$ is completely positive if and only if the 
  matrix $\rho_{\Phi}$ is positive semi-definite.
\end{theorem}

Let us denote the identity operator within $\M{m}$ by $\Id_{m}$.
Consider the map $\Psi_{nm}:\M{n}\to\M{m}$ defined by 
  $\Psi_{nm}(X) = \tr (X) \Id_m$, which is obviously strictly positive.
If $\Phi:\M{n}\to\M{m}$ is an arbitrary positive map 
  then for every $\eps>0$ the map $\Phi+\eps\Psi_{nm}$ is strictly positive. 
This observation leads to the following lemma. 

\begin{lemma}
\label{lem:posimprove}
The set $\{\Phi_\rho\mid\rho\in\DN{mn}\}$ consists of strictly positive maps.
\end{lemma}
\begin{proof}
If $\rho\in\DN{mn}$, then exists $\eps>0$ such that $\rho-\eps \Id_{mn}$ is still 
  positive semi-definite. 
By Choi--Kraus' theorem, $\Phi_{\rho-\eps \Id_{mn}}$ is completely positive, so
  positive.
Using the previous observation we get that
\begin{equation*}
\Phi_\rho = \Phi_{\rho-\eps \Id_{mn}}+\eps\Phi_{\Id_{mn}}
  =\Phi_{\rho-\eps \Id_{mn}}+\eps\Psi_{nm}
\end{equation*}
  is strictly positive.
\end{proof}

The Hilbert--Schmidt scalar product of matrices $X,Y\in\M{n}$ is defined by
\begin{equation*}
\scal{X,Y}=\tr (X^{*}Y).
\end{equation*}
For a linear map $\Phi:\M{n}\to\M{m}$ its adjoint  $\Phi^{*}:\M{m}\to\M{n}$ 
  with respect to the Hilbert--Schmidt scalar product is determined by 
  the equation
\begin{equation*}
\scal{\Phi(X),Y}=\scal{X,\Phi^{*}(Y)}
  \qquad \forall X\in\M{n}\ \forall Y\in\M{m}.
\end{equation*}

\begin{lemma}\label{lem:strpos}
If $\Phi:\M{n}\to\M{m}$ is a strictly positive linear map, then its
  Hilbert--Schmidt adjoint $\Phi^{*}:\M{m}\to\M{n}$ is also strictly positive.
\end{lemma}
\begin{proof}
Let $\Phi:\M{n}\to\M{m}$ be is a strictly positive linear map
  and $B\in\MP{m}\setminus\kz{0}$ be an arbitrary matrix.
Contrary to the lemma, assume that the matrix $\Phi^{*}(B)$ has $\lambda\leq 0$ 
  as eigenvalue with unit length eigenvector $v$.
Consider the orthogonal projection $P=v\scal{v,\cdot}$ onto 
  the one-dimensional subspace generated by $v$. 
Note that because of strict positivity of $\Phi$, we have $\Phi(P)>0$.
Then we get the following contradiction
\begin{align*}
0\geq \lambda&=\tr\gz{\lambda v\scal{v,\cdot}}
  =\tr\gz{\Phi^{*}(B) v\scal{v,\cdot}}=\tr\gz{\Phi^{*}(B)P}\\
&=\tr\gz{P \Phi^{*}(B)}
  =\scal{P,\Phi^{*}(B)}
  =\scal{\Phi(P),B}=\tr(\Phi(P)B)>0.
\end{align*}
\end{proof}

To a matrix $A\in\M{k}$ one can associate the left and right multiplication 
  operators $L_A, R_A:\M{k}\to\M{k}$ that act like
\begin{align*}
X \mapsto L_A (X) &= AX \\
X \mapsto R_A (X) &= XA.
\end{align*}

\begin{lemma}\label{lem:choitranform}
Let $A\in\M{n}$, $B\in\M{m}$ and $\Phi:\M{n}\to\M{m}$ be a linear map. 
Then we have the following identities.
\begin{eqnarray*}
&\rho_{L_B \circ\Phi} = (\Id_n\otimes B)\rho_{\Phi} \qquad
 &\rho_{R_B \circ\Phi} = \rho_{\Phi}(\Id_n\otimes B) \\
&\rho_{\Phi\circ L_A} = (A^{\Tra}\otimes \Id_m)\rho_{\Phi} \qquad 
 &\rho_{\Phi\circ R_A} = \rho_{\Phi}(A^{\Tra}\otimes \Id_m)
\end{eqnarray*}
\end{lemma}
\begin{proof}
These identities can be verified by direct calculations.
\begin{align*}
\rho_{L_B \circ\Phi}&=\sum_{i,j=1}^{n} E_{ij}\otimes B\Phi (E_{ij}) 
  =(\Id_n\otimes B)\sum_{i,j=1}^{n} E_{ij}\otimes \Phi (E_{ij}) 
  =(\Id_n\otimes B)\rho_{\Phi} \\
\rho_{R_B \circ\Phi} &= \sum_{i,j=1}^{n} E_{ij}\otimes \Phi (E_{ij})B 
  =\gz{\sum_{i,j=1}^{n} E_{ij}\otimes \Phi (E_{ij})} (\Id_n\otimes B) 
  =\rho_{\Phi}(\Id_n\otimes B)
\end{align*}
For $A=E_{kl}$, we have
\begin{equation*}
\rho_{\Phi\circ L_A}=\sum_{i,j=1}^{n} E_{ij}\otimes \Phi (E_{kl} E_{ij}) 
  =\sum_{j=1}^{n} E_{lj}\otimes \Phi (E_{kj}) 
  =\sum_{i,j=1}^{n} E_{lk}E_{ij}\otimes \Phi (E_{ij}) 
  =(A^{\Tra}\otimes \Id_m)\rho_{\Phi}
\end{equation*}
and similarly
\begin{equation*}
\rho_{\Phi\circ R_A} 
  =\sum_{i,j=1}^{n} E_{ij}\otimes \Phi (E_{ij}E_{kl}) 
  =\sum_{i=1}^{n} E_{ik}\otimes \Phi (E_{il}) \\
  =\sum_{i,j=1}^{n} E_{ij}E_{lk}\otimes \Phi (E_{ij}) 
  =\rho_{\Phi}(A^{\Tra}\otimes \Id_m).
\end{equation*}
The general case follows from this by the linearity of
  maps $A\mapsto L_A$ and $A\mapsto R_A$.
\end{proof}

For the set of positive semi-definite matrices we have
\begin{enumerate}[i.]
\item $\MP{n}+\MP{n}\subseteq\MP{n}$,
\item $\MP{n}\cap -\MP{n}=\{0\}$,
\item $\lambda\MP{n}\subseteq\MP{n}$ holds for all $\lambda\ge 0$,
\end{enumerate}
  which means $\MP{n}$ is a cone in the real vector space of self-adjoint 
  matrices i.e. $\MP{n}\subseteq\MSa{n}$.
The set $\MP{n}$ is a closed set with nonempty interior.
Now, we introduce the Hilbert metric \cite{Hilbert1895} on this cone.  
In the following we present the definitions and theorems to our settings. 
The general form of these results can be found for example in
  \cite{schr2015,Lemmens2012}.

\begin{definition}
Birkhoff's version of Hilbert's pseudometric is defined on 
  $\MP{n}\setminus\{0\}$ by
\begin{equation*}
d_{\Hilb}(A,B)=\log\gz{\frac{M(A,B)}{m(A,B)}},
\end{equation*}
where
\begin{align*}
M(A,B) &= \inf \kz{\lambda\in\R\mid \lambda B-A\in\MP{n}} \\
m(A,B) &= \sup \kz{\lambda\in\R\mid A-\lambda B\in\MP{n}}.
\end{align*}
\end{definition}

The pseudometric $d_{\Hilb}$ is symmetric, satisfies the triangle 
  inequality, for every $A,B\in\MP{n}\setminus\{0\}$ we have
  $d_{\Hilb}(A,B)=0$ if and only if there exists $c\in\R^{+}$ such that
  $A=cB$ and $d_{\Hilb}$ remains invariant under scaling by 
  positive constant i.e. 
  $d_{\Hilb}(A,B)=d_{\Hilb}(c A,B)=d_{\Hilb}(A,c B)$
  for every $c\in\R^{+}$. 
So $d_{\Hilb}$ is a \emph{projective metric}, it measures the distance of rays 
  and not elements.

Notice that for $A,B\in\MPP{n}$ positive-definite and invertible matrices, 
  one has 
\begin{equation*}
\eig{B^{-1/2}AB^{-1/2}}=\eig{AB^{-1}},
\end{equation*}
  where $\eig{X}$ denotes the set of eigenvalues of $X$. 
Therefore on the space of positive semi-definite matrices, the 
  Birkhoff--Hilbert metric has the following simpler form
\begin{equation*}
d_{\Hilb}(A,B)=\log\gz{\frac{\max (\eig{AB^{-1}})}{\min (\eig{AB^{-1}})}}.
\end{equation*}
Using this form, we can see that matrix inversion is an isometric map 
  on $\MPP{n}$ with respect to $d_{\Hilb}$. 
Similarly, for general positive matrices $A,B\in\MP{n}\setminus\{0\}$ one has
\begin{equation*}
d_{\Hilb}(A,B)=
\begin{cases}
\di \log\gz{
  \frac{\max (\eig{AB^{-1}})}{\min (\eig{AB^{-1}}\setminus\kz{0})}} & 
  \text{if}\ \supp A= \supp B, \\
\di  \infty & \text{if}\ \supp A\neq\supp B,
\end{cases}
\end{equation*}
  where $\supp X$ denotes the orthogonal subspace to $\Ker X$
  and the inverse is denoting the pseudo inverse (inverse on
  the support).
From this form we can see that the pseudo inversion is an isometric map 
  on $\MP{n}\setminus\kz{0}$ with respect to $d_{\Hilb}$.
Moreover, the space  $\MP{n}\setminus\kz{0}$ is complete with respect to
  $d_{\Hilb}$.
We will intensively use these observations in the proof of Theorem 
  \ref{thm:GP}.

\begin{definition}
Given a strictly positive linear mapping $\Phi:\MSa{n}\to\MSa{m}$, the \emph{projective diameter} of $\Phi$ is
\begin{equation*}
\Delta (\Phi) = \sup\kz{d_{\Hilb}(\Phi (A),\Phi (B))\middle| A,B\in\MPP{n}}.
\end{equation*}
Moreover, the \emph{Birkhoff contraction ratio} of $\Phi$ is given by
\begin{equation*}
\delta(\Phi) = \inf\kz{\lambda\in\R^{+}\middle|
  d_{\Hilb}(\Phi(A),\Phi(B))\le\lambda d_H(A,B)\text{ for all } A,B\in\MPP{n}}.
\end{equation*}
\end{definition}

Birkhoff-Hopf theorem on linear maps between cones provides upper bound 
  for the Birkhoff contraction ratio in terms of the projective diameter. 
This classical result was proved first by Birkhoff 
  \cite{Birkhoff1957,Birkhoff1962} and
  similar theorems were discovered by Hopf \cite{Hopf1963,Hopf1963Remark} 
  who was apparently unaware of Birkhoff's work.
It is valid in much more general settings, but the following restricted version 
  will be enough for our purposes.

\begin{theorem}[Birkhoff-Hopf]
\label{thm:Birkhoff} 
If $\Phi:\MSa{n}\to\MSa{m}$ is a strictly positive linear map, then
\begin{equation*}
\delta (\Phi)=\tanh\gz{\frac{1}{4}\Delta (\Phi)},
\end{equation*}
  where $\tanh (\infty)=1$.
\end{theorem}

Finally we cite Sinkhorn's theorem, since our main theorem is a similar 
  result in a similar setting.

\begin{theorem}[Sinkhorn \cite{Sinkhorn1964}]
\label{thm:Sinkhorn}
If $A$ is an $n\times n$ matrix with strictly positive elements, 
  then there exist diagonal matrices $D_{1}$ and $D_{2}$ with 
  strictly positive diagonal elements such that $D_{1}AD_{2}$ is doubly 
  stochastic. 
The matrices $D_{1}$ and $D_{2}$ are unique modulo multiplying the 
  first matrix by a positive number and dividing the second one by the same
  number. 
\end{theorem}

\section{From states to copulas}

We define an equivalence relation among bipartite quantum states 
  such that each equivalence class consists only of either separable 
  or entangled states.

\begin{definition}
We say that $\rho\in\DN{mn}$ and $\rho'\in\DN{mn}$ are \emph{connected} 
  if there exist invertible matrices $A\in\M{n}$ and $B\in\M{m}$ such that
\begin{equation*}
\rho'=(A^{*}\otimes B^{*})\rho (A\otimes B).
\end{equation*} 
Connectivity of states $\rho$ and $\rho'$ is denoted by
  $\rho\leftrightsquigarrow\rho'$.
\end{definition}

Easy to check the following observation about the connecticity of states.

\begin{lemma}
Connected states are mutually separable or entangled.
\end{lemma}

Now let us define the first candidates to copulas, namely
  those composite states which have uniform marginals.

\begin{definition}
A state $\rho\in\DN{nm}$ is called as a precopula if it has uniform 
  marginals.
The set of precopulas is denoted by $\Ch'_{mn}$, that is
\begin{equation*}
\Ch'_{mn}=\kz{\rho\in\DN{mn}\middle| \tr_1 (\rho)=\frac{1}{m}\Id_m,\,
  \tr_2 (\rho)=\frac{1}{n}\Id_n}.
\end{equation*}
For states $\rho,\rho'\in\Ch'_{mn}$ we write $\rho\sim\rho'$ if and only if 
  there exist unitaries $U\in\M{n}$ and $V\in\M{m}$ such that
\begin{equation*}
 \rho'=(U^{*}\otimes V^{*})\rho (U\otimes V).
\end{equation*}  
Since, $\rho$ is an equivalence relation on $\Ch'_{mn}$ 
  we can take the quotient space
\begin{equation*}
\Ch_{mn}=\bigslant{\Ch'_{mn}}{\sim}.
\end{equation*}
Elements of the quotient space $\Ch_{mn}$ are called \emph{copulas}.
\end{definition}
 
To prove our main result we use the fact that the marginals of a state
  can be obtained as applying the quantum channel and its adjoint
  defined by the state to the identity matrix.

\begin{lemma}
\label{lem:chn} 
A density matrix $\rho\in\DN{nm}$ has marginals
  $\rho^{(1)}\in\DN{n}$ and $\rho^{(2)}\in\DN{m}$ if and only if
\begin{align*}
\Phi_{\rho} (\Id_n)    &= \rho^{(2)} \\
\Phi_{\rho}^{*}(\Id_m) &= \rho^{(1)}.
\end{align*}
\end{lemma}
\begin{proof}
The second marginal can be expressed as
\begin{align*}
\rho^{(2)}=\tr_1 \rho = \sum_{i,j=1}^n \tr (E_{ij})\Phi_{\rho}(E_{ij}) 
  =\sum_{i=1}^n\Phi_{\rho}(E_{ii}) = \Phi_{\rho} (\Id_n).
\end{align*}
For the first marginal, we have
\begin{align*}
\rho^{(1)}=\tr_2 \rho=\sum_{i,j=1}^n E_{ij}\tr\gz{\Phi_{\rho}(E_{ij})}
  =\sum_{i,j=1}^n E_{ij} \tr\gz{\Phi_{\rho}^{*}(\Id_m)E_{ij}},
\end{align*}
  where the last expression is nothing else but the matrix 
  $\Phi_{\rho}^{*}(\Id_m)$ in the basis $(E_{ij})_{i,j=1,\dots,n}$, which
  verifies the equality $\rho^{(1)}=\Phi_{\rho}^{*}(\Id_m)$.
\end{proof}

According to the previous lemma the density matrix $\rho\in\DN{nm}$
  is a precopula if and only if $\Phi_\rho (\Id_n)=\frac{1}{m}\Id_{m}$
  and  $\Phi^{*}_{\rho}(\Id_m)=\frac{1}{n}\Id_{n}$.
The next theorem is about a kind of fixed point property of strictly
  positive maps.
We have modified and adopted to our settings the theorem which appeared 
  first in \cite{schr2015}. 

\begin{theorem}
\label{thm:GP} 
For any strictly positive linear map $\Phi:\M{n}\to\M{m}$ there exist unique
  matrices $\varphi_0\in\MPP{n}$ and $\varphi_1\in\MPP{m}$
  up to a positive multiplicative constant such that
\begin{align*}
\Phi (\varphi_0^{-1}) &= \frac{1}{m} \varphi_1^{-1} \\
\Phi^{*}(\varphi_1) &= \frac{1}{n} \varphi_0.
\end{align*}
\end{theorem}
\begin{proof}
We mimic the proof of Theorem 6 in \cite{schr2015}. 
Let us define the matrix inversion as function $i:\MPP{}\to\MPP{}$,
  $i(\rho)=\rho^{-1}$ for every $\rho\in\MPP{}$.
Consider the map $T:\MP{n}\to\MP{n}$ defined as
  $T=i\circ\Phi^{*}\circ i\circ\Phi$.

We show that $T$ is a contraction with respect to the Hilbert metric.
As it was mentioned the matrix inversion is an isometry. 
If $\Phi:\MP{n}\to\MP{m}$ is a contraction then  by Lemma 
  \ref{lem:strpos} the same holds for $\Phi^{*}:\MP{m}\to\MP{n}$,
  so in this case $T$ is a composition of contractions and isometries.

To prove that $\Phi$ is a contraction we estimate its projective diameter. 
Since the Hilbert metric is invariant under scaling by positive scalars 
  we can restrict our attention to $\overline{\DN{n}}$ which leads to the following estimation.
\begin{equation*}
\Delta (\Phi)=\sup\kz{d_{\Hilb}(\Phi(A),\Phi(B))\middle| A,B\in\MPP{n}} 
  =\sup\kz{d_{\Hilb}(\Phi(\rho),\Phi(\rho'))\middle| 
  \rho,\rho'\in\overline{\DN{n}}}
\end{equation*}
The map $\Phi:\M{n}\to\M{m}$ is strictly positive thus it sends all the states 
  in $\overline{\DN{n}}$ to $\MPP{m}$, so the map 
  $(\rho,\rho')\mapsto d_{\Hilb}(\Phi (\rho), \Phi (\rho'))$ is continuous 
  on the compact set $\overline{\DN{n}}\times\overline{\DN{n}}$. 
This implies that $\Delta (m\Phi)<\infty$ and by Theorem \ref{thm:Birkhoff}, 
  we have $\delta (\Phi)=\tanh \gz{\frac{1}{4}\Delta(\Phi)}<1$, which
  means that $\Phi$ is a contraction.

By Banach fixed-point theorem we can conclude that there exists a unique fixed
 ray given by the state $\varphi\in\DN{n}$ such that 
  $T(\varphi)=\lambda\varphi$ for some $\lambda>0$. 

Let us define $\varphi_{1}=\frac{1}{m} (i\circ\Phi)(\varphi)$ and
  $\varphi_{0}=n\Phi^{*}(\varphi_1)$. 
In this case
\begin{equation*}
\varphi_{0}^{-1}=i(\varphi_{0})
  =\frac{m}{n}(i\circ\Phi^{*}\circ i\circ\Phi)(\varphi)
  =\frac{m}{n}T(\varphi)=\frac{\lambda m}{n}\varphi,
\end{equation*}
  therefore
\begin{equation*}
\Phi(\varphi_{0}^{-1})=\frac{\lambda m}{n}\Phi(\varphi)
  =\frac{\lambda}{n}\varphi_{1}^{-1}.
\end{equation*}
Now we can determine the value of the parameter $\lambda$ from the
  following equation.
\begin{align*}
n &=\tr (\Id_n)=\tr(\varphi_0^{-1}\varphi_0)
  =\tr(\varphi_0^{-1}n\Phi^{*}(\varphi_1))
  =n\scal{\varphi_0^{-1},\Phi^{*}(\varphi_1)}\\
&=n\scal{\Phi(\varphi_0^{-1}),\varphi_1}
  =n\tr\Phi(\varphi_0^{-1}) \varphi_1 
  =n\tr \frac{\lambda}{n}\varphi_{1}^{-1} \varphi_1
=\lambda m
\end{align*}
Since $\lambda=\frac{n}{m}$, the equations
  $\Phi (\varphi_0^{-1})=\frac{1}{m} \varphi_1^{-1}$ and
  $\Phi^{*}(\varphi_1)=\frac{1}{n}\varphi_0$ are fulfilled.
\end{proof}

The next theorem formulates the main result of this paper. 
It establishes a separability-preserving transformation in composite systems
  between states and copulas.
We prove that each state in $\DN{mn}$ is connected to a well defined copula. 
This leads to the conclusion that copulas describe the dependence between the 
  subsystems. 
We use the language of completely positive maps and the Choi--Jamio\l{}kowski 
  isomorphism to prove this.

\begin{theorem} 
For any element $\rho\in\DN{mn}$, there exists a unique copula 
  $\tilde{\chi}\in\Ch_{mn}$ such that for any representative $\chi$ 
  of $\tilde{\chi}$ one has $\rho\leftrightsquigarrow\chi$.
\end{theorem}
\begin{proof}
Let $\rho\in\DN{mn}$ be arbitrary. 
By Lemma \ref{lem:posimprove}, $\Phi_{\rho}:\M{n}\to\M{m}$ is 
  strictly positive. 
Theorem \ref{thm:GP} guarantees that there exist unique matrices
  $\varphi_0\in\MPP{n}$ and $\varphi_1\in\MPP{m}$ up to a positive 
  multiplicative constant such that
\begin{align*}
\Phi_{\rho} (\varphi_0^{-1}) &= \frac{1}{m} \varphi_1^{-1} \\
\Phi_{\rho}^{*}(\varphi_1) &= \frac{1}{n} \varphi_0.
\end{align*}
Taking any factorization $\varphi_0=\psi_0^{*}\psi_0$, 
  $\varphi_1=\psi_1^{*}\psi_1$, for the map 
\begin{equation*}
\Theta:\M{n}\to\M{m}\qquad
X\mapsto   \psi_1\Phi_{\rho}\gz{\psi_0^{-1} X \gz{\psi_0^{-1}}^{*}}\psi_1^{*}
\end{equation*}
we have
\begin{align*}
\Theta (\Id_n) &= \frac{1}{m}\Id_m \\
\Theta^{*} (\Id_m) &= \frac{1}{n}\Id_n.
\end{align*}
Using Lemma \ref{lem:chn} we conclude that the state $\rho_{\Theta}\in\DN{mn}$
   has uniform marginals, so $\rho_{\Theta}\in\Ch_{mn}'$.
According to Lemma \ref{lem:choitranform},
\begin{equation*}
\rho_{\Theta} = \gz{\gz{\psi_0^{-1}}^{\Tra}\otimes\psi_1}
\rho\gz{\gz{\psi_0^{-1}}^{\Tra}\otimes\psi_1}^{*}.
\end{equation*}

On the other hand, factorization of $\varphi_0$ and $\varphi_1$ is unique up to 
  unitary equivalence. 
If $\varphi_0=\psi_0'^{*}\psi_0'$, $\varphi_1=\psi_1'^{*}\psi_1'$, 
  then there exists unitaries $U_0\in\M{n}$ and $U_1\in\M{m}$ such that 
  $\psi_0'=U_0\psi_0$ and $\psi_1'=U_1\psi_1$. 
The composite state corresponding to the transform defined by $\varphi_0'$ 
  and $\varphi_1'$ is
\begin{align*}
\gz{\gz{\psi_0'^{-1}}^{\Tra}\otimes\psi_1'}\rho
\gz{\gz{\psi_0'^{-1}}^{\Tra}\otimes\psi_1'}^{*} 
=\gz{\gz{U_0^{-1}}^{\Tra}\otimes U_1}\rho_\Theta 
\gz{\gz{U_0^{-1}}^{\Tra}\otimes U_1}^{*}
\end{align*}
  and therefore represents the same copula as $\rho_{\Theta}$.
\end{proof}

The proof of the previous theorem shows that for any state $\rho\in\DN{nm}$ 
  there are unique matrices $A$ and $B$ up to multiplication by unitaries 
  such that $(A\otimes B)\rho(A\otimes B)^{*}$ is a precopula.
Similar phenomena can be discovered in Sinkhorn's theorem 
  (Theorem \ref{thm:Sinkhorn}), where for a 
  matrix $D$ with strictly positive elements there exist diagonal
  matrices $A$ and $B$ up to a positive multiplicative constant
  (i. e. one can take the matrices $cA$ and $\frac{1}{c}B$ for $c\in\R^{+}$)
  such that the matrix $ADB$ is doubly stochastic.

\section{Conclusions}

We introduced the notion of copula as an equivalence class of those states 
  which have uniform marginals and we proved that for every state there
  exists a unique copula connected to the state.
The construction of the copula of a given state uses Banach fixed-point 
  theorem, so we do not have explicit formula for the corresponding copula, 
  but numerically can be computed with arbitrary precision.
We have found very fast convergence for randomly generated qubit-qubit states.

The following are examples for connecting open question to copulas.
How one can reconstruct a state from its copula and its marginals?
How one can parametrize the space of copulas? 
Which quantities are hereditary from states to copulas, namely
  for given states $D_{1},D_{2}$ and the corresponding copulas 
  $\chi_{1},\chi_{2}$ for which quantity $H$ holds the equality
  $H(D_{1},D_{2})=H(\chi_{1},\chi_{2})$.
Is it true for some kind of entropy or distance measure?

\section*{Acknowledgement}

We wish to thank P\'eter Vrana for his contribution to this project, 
  drawing our attention to Sinkhorn's theorem and for inspiring personal 
  discussions. Authors are also grateful to Michael Hall for bringing their attention
  to the usefulness of copula in the theory of quantum steering. 

Both of the authors enjoyed the support of the 
  Hungarian National Research, Development and Innovation Office (NKFIH) grant 
  no. K124152 and the Hungarian Academy of Sciences Lend\"ulet-Momentum grant 
  for Quantum Information Theory, no. 96 141. 
The first author was also supported by the Lend\"ulet grant LP 2015-6 of the 
  Hungarian Academy of Sciences.

\bibliographystyle{plain} 
%\bibliography{copula}    

\begin{thebibliography}{10}

\bibitem{Birkhoff1957}
G.~Birkhoff.
\newblock Extensions of {Jentzsch's} theorem.
\newblock {\em Trans. Amer. Math. Soc.}, 85(1):219--27, 1957.

\bibitem{Birkhoff1962}
G.~Birkhoff.
\newblock Uniformly semi-primitive multiplicative processes.
\newblock {\em Trans. Amer. Math. Soc.}, 104:37--51, 1962.

\bibitem{Bowles2016}
J.~Bowles, F.~Hirsch, M.~T.~Quintino and N.~Brunner.
\newblock Sufficient criterion for guaranteeing that a two-qubit state is unsteerable.
\newblock {\em Phys. Rev. A}, 2(93):022121, 2016.

\bibitem{Choi}
M.~D. Choi.
\newblock Completely positive linear maps on complex matrices.
\newblock {\em Linear Algebra and Appl.}, 10:285--290, 1975.

\bibitem{Frechet1951}
M.~Fr\'echet.
\newblock Sur les tableaux de corr\'elation dont les marges sont donn\'ees.
\newblock {\em Ann. Univ. Lyon Sc.}, 4:53--84, 1951.

\bibitem{schr2015}
T.~T. Georgiou and M.~Pavon.
\newblock Positive contraction mappings for classical and quantum
  {Schr{\"o}dinger} systems.
\newblock {\em J. Math. Phys.}, 56(3):033301, 2015.

\bibitem{Hilbert1895}
D.~Hilbert.
\newblock Ueber die gerade linie als k{\"u}rzeste verbindung zweier punkte.
\newblock {\em Math. Ann.}, 46:91--96, 1895.

\bibitem{Hopf1963}
E.~Hopf.
\newblock An inequality for positive linear integral operators.
\newblock {\em J. Math. Mech.}, 12(5):683--692, 1963.

\bibitem{Hopf1963Remark}
E.~Hopf.
\newblock Remarks on my paper {"}an inequality for positive linear integral
  operators{"}.
\newblock {\em J. Math. Mech.}, 12(6):889--892, 1963.

\bibitem{Jamiolkowski}
A.~Jamio{\l k}owski.
\newblock Linear transformations which preserve trace and positive
  semidefiniteness of operators.
\newblock {\em Rep. Mathematical Phys.}, 3(4):275--278, 1972.

\bibitem{Lemmens2012}
B.~Lemmens and R.~Nussbaum.
\newblock {\em Nonlinear Perron-Frobenius Theory}.
\newblock Cambridge Tracts in Mathematics. Cambridge University Press, 2012.

\bibitem{Milne2014}
A.~Milne, S.~Jevtic, D.~Jennings, H.~Wiseman and T.~Rudolph.
\newblock {Quantum steering ellipsoids, extremal physical states and monogamy}.
\newblock {\em New J. of Phys.}, 8(16):083017, 2014.

\bibitem{NielsenChuang}
M.~A. Neilsen and I.~L. Chuang.
\newblock {\em Quantum Computation and Quantum Information}.
\newblock Cambridge University Press, Cambridge, 2000.

\bibitem{Nelsen2006}
R.~B. Nelsen.
\newblock {\em An introduction to copulas.}
\newblock Springer, New York, 2006.

\bibitem{Neumann1932}
J.~von Neumann.
\newblock {\em {Mathematische Grundlagen der Quantenmechanik}. ({German})
  [{Mathematical} Foundations of Quantum Mechanics]}.
\newblock Springer-Verlag, Berlin, Germany~/ Heidelberg, Germany~/ London, UK~/
  etc., 1932.

\bibitem{Petz2008}
D.~Petz.
\newblock {\em Quantum information theory and quantum statistics}.
\newblock Theoretical and Mathematical Physics. Springer-Verlag, Berlin, 2008.

\bibitem{Sinkhorn1964}
R.~Sinkhorn.
\newblock A relationship between arbitrary positive matrices and doubly
  stochastic matrices.
\newblock {\em Ann. Math. Statist.}, 35(2):876--879, 1964.

\bibitem{shuming2016}
C.~Shuming, M.~Anthony, J.~W.-H.~Michael, H.~M.~Wiseman.
\newblock Volume monogamy of quantum steering ellipsoids for multiqubit systems.
\newblock {\em Phys. Rev. A}, 94(4):042105, 2016.

\bibitem{Sklar1959}
A.~Sklar.
\newblock Fonctions de r{\'e}partition {\`a} n dimensions et leurs marges.
\newblock {\em Publications de l'Institut de Statistique de l'Universit\'e de
  Paris}, 8:229--231, 1959.
\end{thebibliography}

% Copy from "copula01.bbl".

% End of copy.

\end{document}